\documentclass[a4paper,10pt]{article}
\usepackage{amsthm,amsmath,amssymb}
\usepackage{graphicx}
\usepackage[latin1]{inputenc}
\usepackage[T1]{fontenc}
\usepackage[norsk,english]{babel}

\usepackage{color}

\usepackage[affil-it]{authblk}

\newtheorem{theorem}{Theorem}

\newtheorem{lemma}{Lemma}

\newtheorem{definition}{Definition}

\theoremstyle{definition}
  \newtheorem{observation}{Observation}

\theoremstyle{definition}
  \newtheorem{example}{Example}

\newcommand{\alf}{\ensuremath{\mathcal{A}}}
\newcommand{\N}{\ensuremath{\mathbb{N}}}
\newcommand{\Z}{\ensuremath{\mathbb{Z}}}

\newcommand{\R}{\ensuremath{\mathbb{R}}}

\newcommand{\limn}{\ensuremath{\lim_{n\rightarrow\infty}}}      

\newcommand{\M}{\ensuremath{\mathcal{M}}}

\begin{document}

\title{Expressing the entropy of lattice systems as sums of conditional entropies.}

\author{Torbj{\o}rn Helvik \;and Kristian Lindgren\footnote{kristian.lindgren@chalmers.se}}

\affil{Complex systems group, Department of Energy and Environment, \\ Chalmers University of Technology, SE-41296 G\"oteborg, Sweden}

\date{\today}

\maketitle

\begin{abstract}
Whether a system is to be considered complex or not depends on how one searches for correlations. We propose a general scheme for calculation of entropies in lattice systems that has high flexibility in how correlations are successively taken into account. 
Compared to the traditional approach for estimating the entropy density, in which successive approximations builds on step-wise extensions of blocks of symbols, we show that one can take larger steps when collecting the statistics necessary to calculate the entropy density of the system. In one dimension this means that, instead of a single sweep over the system in which states are read sequentially, one take several sweeps with larger steps so that eventually the whole lattice is covered. This means that the information in correlations is captured in a different way, and in some situations this will lead to a considerably much faster convergence of the entropy density estimate as a function of the size of the configurations used in the estimate. The formalism is exemplified with both an example of a free energy minimisation scheme for the two-dimensional Ising model, and an example of increasingly complex spatial correlations generated by the time evolution of elementary cellular automaton rule 60.

 
\end{abstract}

\section{Introduction}

Many models in statistical mechanics involve a lattice of particles having spins or other states from a finite set, and with
 interaction between neighboring particles defined by a transition invariant potential. The Ising model, which was solved
by Onsager in 1944 \cite{onsager}, epitomizes such models. An important problem is to find the entropy density of the Gibbs
state corresponding to the interaction potential. 

This and other properties of lattice models are often investigated by Monte Carlo relaxation methods such as the Metropolis
algorithm. These methods yields estimates of thermodynamic quantities that can be directly measured in simulation runs, such
as internal energy $u$ and long range order. However, neither the entropy density $s$ nor the free energy $f = u-Ts$ can be obtained
directly, so special approaches are needed. Several methods have been devised to estimate the entropy of lattice systems using MC
simulations, the most important one being thermodynamic integration. See \cite{binder85} for a review. We will in this
paper be concerned with how the entropy density can be written in terms of sums of appropriate conditional entropies of spin
variables. The method is not confined to finding entropies of Gibbs states, but can be applied to any probability measure on
lattice systems in any dimension.

The first person which, to our knowledge, used conditional entropies to investigate two-dimensional lattice systems was Alexandrowicz
\cite{alexandrowicz71}.
His approach was to generate lattice configurations by adding spins one by one to an empty lattice using a Markov process. The
transition probabilities of the process was parameterized and depended on previous spins in some neighborhood. The best
parameter values was found for each inverse temperature $\beta$ by 
optimizing according to the minimum free energy principle. Entropy was then estimated as the average of $\log\frac{1}{p}$, with $p$
being the probability of the realized transition. This is tantamount to estimating the conditional entropy of a spin with
respect to the spins in some neighborhood.

In \cite{meirovitch77} Meirovitch introduced the idea that instead of searching for optimal transition probabilities, the transition
probabilities could be directly estimated from looking at frequencies in a lattice configuration obtained, e.g., from a
Monte Carlo algorithm. In \cite{schlijper89,schlijper90} Schlijper et.~al. combined this method of calculating entropy with the
Cluster Variation Method \cite{schlijper83} to obtain both a lower bound and an upper bound on the entropy. They also put the method
on formal ground, in particular using a result on the global Markov property for spin systems \cite{goldstein90}.

The method of using empirical frequencies for estimating the entropy has several advantages. It is cost effective and can
easily be included into a MC algorithm to monitor the entropy and free energy during a MC simulation. The method basically needs
only a single lattice configuration and is easily adaptable to more involved problems. Schlijper et.~al., and also Alexandrowicz,
pointed out that for a lattice configuration fluctuations in energy and entropy tend to cancel. As a consequence of this, free
energy, which often is the interesting quantity, is more easily determined than its separate contributions entropy and energy. The
effectivity of the method has been demonstrated by both Meirovitch and Schlijper et.~al, and it has been used in several applications,
e.g., \cite{Marcelja,Meirovitch83,Meirovitch83b,Meirovitch99,Olbrich00}.

An approach based on conditional entropies usually means that one searches for correlations, and the more information that is found in correlations the less is the estimate of the entropy. This search for correlations typically involves extensions of blocks in a regular way: In one dimension, one uses  entropies conditioned on an increasing sequence of lattice sites to the left. In two dimensions, one may extend a rectangular block in a similar way forming conditional entropies, or, alternatively one may extend blocks one lattice site at the time using a lexicographic ordering, as was exploited already by Kramers and Wannier \cite{kramers41}. How well this approach works for estimating the entropy depends on convergence properties of the conditional entropies, reflecting how long correlations that are present in the system and how they decline with distance. In a complex system, correlations may not be so easily detected, and it may turn out that the traditional approach, extending blocks by adding neighbouring lattice sites, is not efficient. Instead one may consider a search for correlations in which larger steps are taken, temporally disregarding states in lattice sites in-between, which is an approach we present here.

The main contribution of this paper is a general and flexible scheme for obtaining representations of the entropy density of a
lattice system in terms of suitable sums of conditional entropies. This is achieved by scanning the lattice in the order suggested by
some regular sublattice, and possibly multiple times in succession. Our results adds flexibility to the empirical
frequency approach to estimating entropy. We present the procedure for general point lattices in Section \ref{entropy_densities}. In
Section \ref{Gibbs_states}, we discuss how the procedure can be helpful in finding simple expressions when the measure is a Gibbs state with finite
interaction range, and this is exemplified by an entropy estimate applied to the two-dimensional Ising model. In Section \ref{sec:CA}, we show how the main theorem can be used to create a hierarchical decomposition of the entropy, and this is illustrated with entropy estimates of patterns generated by elementary cellular automaton rule 60.

The result on representations of entropy densities are also related to several information theoretic notions of structure and
complexity in lattice systems. This includes the effective measure complexity introduced by Grassberger \cite{grassberger86}, the excess entropy \cite{crutchfield03}, and local information introduced for one dimensional systems in \cite{helvik-et-al}. Whether a system is considered complex or not depends on how one searches for correlations -- a clever scheme for scanning the lattice may reveal structure that would not be so easily detected using a traditional approach.

\section{Entropy density in Point lattices}
\label{entropy_densities}

\subsection{Background}

By the term lattice we mean a regularly spaced array of points in $\R^d$. The correct mathematical term is \emph{point lattice}, as
the term lattice is a more general structure. Formally, a point lattice $V$ is a a discrete abelian subgroup of $\R^d$. Typical
point lattices in $\R^2$ are the quadratic grid and the hexagonal lattice. The triangular lattice is actually a union of two point
lattices which are translates of each other by a constant. We will use the term lattice in this paper to refer to the union of a
finite number translations of a single point lattice $V$. We will first present our result, Theorem \ref{teo:mainteo}, for a single
lattice. We then discuss how it can be extended to a general union of translates of a point lattice.

When drawing a lattice $V\subset \R^d$ it is convenient to draw the corresponding Voronoi cells instead. The Voronoi cell of $v\in
V$ is the subset $\{x\in\R^d : |x-v|<|x-u|\,\forall\, u\in V, v\neq u \}$. The collection of all Voronoi cells comprises a
periodic tiling of $\R^d$. Se Fig.~\ref{fig:2D_lattices} for illustrations.

\begin{figure}
\centering
\includegraphics[width=\textwidth]{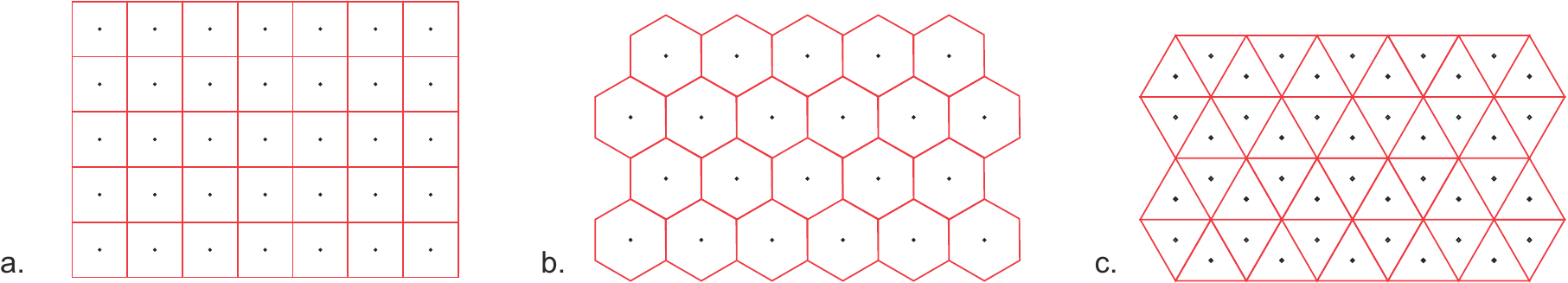}
\caption{Lattices in $\R^2$ with corresponding Voronoi cells.
\textbf{a.} Cubic, \textbf{b.} Hexagonal, \textbf{c.} Triangular.}
\label{fig:2D_lattices}
\end{figure}

We now consider spin systems on a point lattice $V$. By this we mean that each lattice point represents a particle which has some
spin from a finite set $\alf$. We first introduce some terminology. Let $V$ be a point lattice. For guiding the
intuition, it can be useful to think of $V$ as $\Z^d$ or just the  grid $\Z^2$. The point lattice is represented by an ordered
collection $e_1,\dots,e_d $ of linearly independent unit vectors in $\R^d$ such that
\[
V = \langle e_1,\dots,e_d\rangle = \left\{ \sum_{i=1}^d v_i e_i :
v_i\in\Z\right\}.
\]
The representation of a point $v\in V$ in terms of $(v_1,\dots,v_d)$ is unique. Note that $\underline{0}\in V$ for any point lattice $V$. For two subsets $\Lambda,\Lambda'\subset V$ of the point lattice, write
\[
\Lambda+\Lambda' = \{v+v' : v\in \Lambda, v'\in \Lambda'\} \,.
\]
Define the unit $d$-cube as $U_d = \{v\in V : \sup_{1\leq i \leq d} |v_i|\leq 1 \}$, and define the boundary of a set
$\Lambda\subset V$ as
\[
\partial \Lambda = \Lambda \cap (U_d + \Lambda^C) \;.
\]
We say that $\Lambda_n \uparrow V$ in the van Hove sense if
\begin{enumerate}
    \item $\bigcup_{n} \Lambda_n = V$,
    \item $\Lambda_{n}\subset\Lambda_{n+1}$ $\forall n$,
    \item $\limn |\partial \Lambda_n|/|\Lambda_n|= 0 $.
\end{enumerate}

An $m$-dimensional subgroup $G = \langle \hat{e}_1,\dots,\hat{e}_m \rangle$ of the lattice is the collection $\{\sum_i g_i\hat{e}_i
\}$ of all linear combinations of a set of $m$ linearly independent vectors $\hat{e}_1,\dots,\hat{e}_m\in V$. The subgroup
is itself a point lattice. We define a tiling\footnote{This is a non-standard use of the term tiling.} of $V$ as follows.
\begin{definition}
A \emph{tiling} of a $d$-dimensional point lattice $V$ is a pair $(A,G)$
consisting of a finite subset $A\subset V$ with $\underline{0}\in
A$ and a $d$-dimensional subgroup $G$ of $V$ satisfying
\begin{enumerate}
    \item $A+G = V$.
    \item $(A+g) \bigcap (A+h) = \emptyset$ for all $g,h\in G$, $g\neq
    h$.
\end{enumerate}
\end{definition}

We will use the \emph{lexicographical} order on $G$. For $g,g'\in G$, write $g = \sum_i g_i\hat{e}_i$ and $g' = \sum_i
g_i'\hat{e}_i$. 
We say that $g<g'$ if there is an $i$, $1\leq i \leq k$, such that $g_i < g'_i$ and  $g_j = g'_j$ for $j<i$. 
Note that due to the use of the lexicographical order, we do not consider, e.g., the tiling of $V=\Z^2$ with $A=\{(0,0)\}$ and
$G=\langle (1,0),(0,1)\rangle$ and the tiling with $A=\{(0,0)\}$ and $G=\langle (1,0),(1,1)\rangle$ as equal. 

Based on the lexicographical order, we define the subset
$G^- \subset G$ as
\[
G^- = \{g\in G: g<0\} .
\]

An example of a tiling of the $2$-dimensional square lattice is illustrated in Fig.~(\ref{fig:2D_tiling}). 

\begin{figure}
\centering
\includegraphics[width=0.3\textwidth]{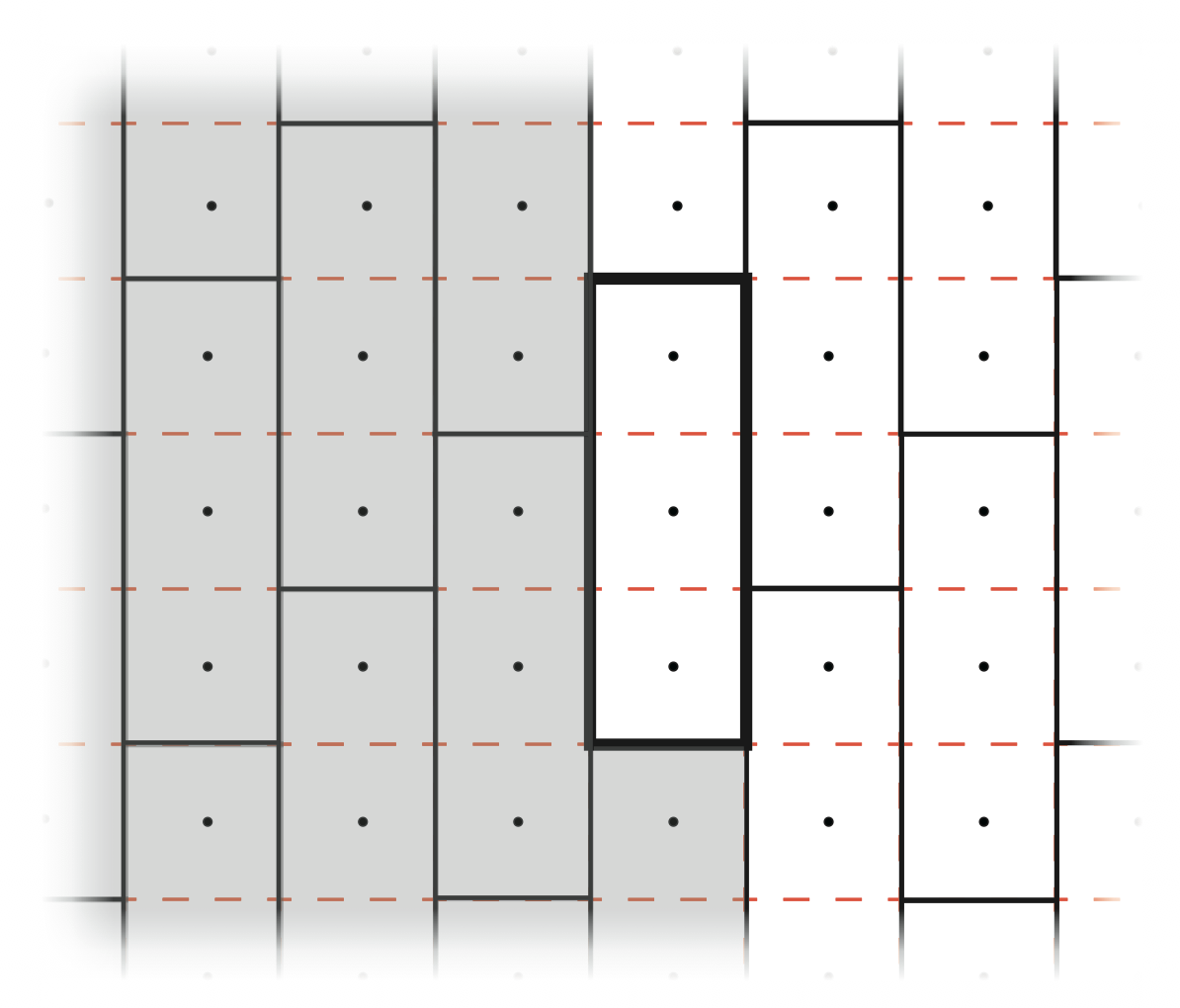}
\caption{Example of a tiling of the square lattice. The tiling is defined by $(A,G)$ with $A=\{(0,0),(0,1),(0,2)\}$ and $G=\langle (1,1),(1,-2) \rangle$. The set $A$ is marked with a thick line and the sub lattice $G^-$ is shaded.}
\label{fig:2D_tiling}
\end{figure}

\subsection{Measures and entropy}

Let $\alf$ be a finite set consisting of the possible spins and let $\Lambda\subseteq V$. An assignment of a spin in $\alf$ to
each element of $\Lambda$ is called a configuration of $\Lambda$. We denote a configuration of $\Lambda$ by $x_{\Lambda}$. The set of all configurations of $\Lambda$ is denoted by $\alf^{\Lambda}$. A configuration $x_{V}$ on the entire space will be denoted just
by $x$.

Let the set $\M$ consist of all translation invariant probability measures on $\alf^{V}$. These are often called states in
statistical mechanics. The restriction of a measure $\mu\in\M$ to $\alf^{\Lambda}$ is denoted by $\mu_{\Lambda}$. We drop the
subscript when no confusion can occur. A translation invariant measure is uniquely defined by specifying its restriction to all
finite $\Lambda\in V$.

The entropy of a subset $\Lambda$ with respect to a measure $\mu$ is defined as
\begin{equation}\label{eq:finiteentropy}
    S(\Lambda) = -\sum_{x_{\Lambda}\in\alf^{\Lambda}}
    \mu_{\Lambda}(x_{\Lambda})\log \mu_{\Lambda}(x_{\Lambda}) \,.
\end{equation}
The entropy is the average information that is gained by observing the configuration on $\Lambda$, where information is used in the
sense of Shannon \cite{Shannon48}.

The entropy density of a measure $\mu$ is defined as the average entropy per spin. That is, as the limit
\begin{equation}
s_\mu = \limn \frac{1}{|\Lambda_n|}S(\Lambda_n),
\end{equation}
where $\Lambda_n \uparrow V$ in the van Hove sense.

It is well known that $s$ also can be written as a \emph{conditional entropy}. Let $\Lambda$ and $\Lambda'$ be finite subsets of $V$. The conditional entropy of $\Lambda$ given $\Lambda'$ with respect to $\mu$ is defined as
\begin{equation}\label{eq:finiteentropy}
    S(\Lambda|\Lambda') = -\sum_{x_{\Lambda\cup \Lambda'}}
    \mu(x_{\Lambda\cup \Lambda'})\log
    \mu(x_{\Lambda}|x_{\Lambda'}),
\end{equation}
where $\mu(x_{\Lambda}|x_{\Lambda'}) =
\mu(x_{\Lambda\cup\Lambda'})/\mu(x_{\Lambda})$. This is the information gained from observing the configuration on $\Lambda$
when the configuration on $\Lambda'$ is known.

For $\Lambda'$ infinite, the conditional entropy is defined by a limit. Define $V_n = \{ \sum_i v_ie_i : |v_i|\leq n\,\forall\,
i\}$, and
\begin{equation}\label{eq:infinitecondentropy}
    S(\Lambda|\Lambda') =\limn S(\Lambda|\Lambda'\cap V_n) \;.
\end{equation}
Convergence is ensured by the monotonicity property of conditional entropy:
\begin{lemma}\label{lem:G_nkonv}
\[
\Lambda' \subseteq \Lambda''\;\Rightarrow\; S(\Lambda|\Lambda'')
\leq S(\Lambda|\Lambda') \;.
\]
\end{lemma}
\begin{proof}
When $\Lambda' \subseteq \Lambda''$ we have
\begin{align}
S(\Lambda|\Lambda') - S(\Lambda|\Lambda'') &= \sum_{x_{\Lambda\cup \Lambda''}}
    \mu(x_{\Lambda\cup \Lambda''})\log \frac{\mu(x_{\Lambda}|x_{\Lambda''})}{\mu(x_{\Lambda}|x_{\Lambda'})} = \nonumber \\
    &= \sum_{x_{\Lambda'' \setminus \Lambda}} \mu(x_{\Lambda'' \setminus \Lambda}) \sum_{x_{\Lambda}} 
    \mu(x_{\Lambda}|x_{\Lambda''}) \log \frac{\mu(x_{\Lambda}|x_{\Lambda''})}{\mu(x_{\Lambda}|x_{\Lambda'})} \ge 0 \,, \nonumber
\end{align}
where the inequality follows from the fact that the second sum is the Kullback-Leibler divergence, or the relative entropy, between the distributions $\mu(x_{\Lambda}|x_{\Lambda'})$ and $\mu(x_{\Lambda}|x_{\Lambda''})$, which is a non-negative quantity \cite{kullback51}. This concludes the proof. 
\end{proof}
\noindent
A further consequence of this result is that conditional entropy is bounded above by $\log |\alf|$ per spin.
\begin{equation}
S(\Lambda|\Lambda') \leq S(\Lambda) \leq |\Lambda|\log |\alf| \;.
\end{equation}
For a one-dimensional lattice it is easy to prove that $s_\mu = S(\{0\}|\{i:i<0\})$ \cite{coverthomas}. It is practical to
represent the expression graphically in the following way
\begin{equation}
s_\mu \hspace{2mm}= \hspace{2mm}
    \includegraphics[width=3.25cm]{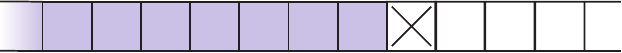}
\end{equation}
The right hand side is to be interpreted as the conditional entropy of the spin at the cross conditioned on the spins all filled cells
(which here are all spins to the left). This representation of the expression will prove useful later, when the situation is more
involved.

Before we proceed to the main result we state a simple but very useful property of conditional entropies that follows from
\eqref{eq:finiteentropy}.
\begin{observation}
For finite $\Lambda, \Lambda', \Lambda'' \subset V$,
\begin{equation}\label{eq:finitecondentropy}
    S(\Lambda\cup\Lambda'|\Lambda'') = S(\Lambda|\Lambda'\cup
    \Lambda'') + S(\Lambda'|\Lambda'') \;.
\end{equation}
\end{observation}

\noindent
Also note that for a translation invariant measure $\mu$ on $\alf^V$,
\[
S(\Lambda|\Lambda') = S(\Lambda+v|\Lambda'+v) \quad \forall \;
v\in V \;.
\]

\subsection{Main result}

The main result of this paper shows that there is great flexibility in choosing how to express the entropy density in
terms of conditional entropies. The various expressions are achieved through using different tilings of $V$ and coverings of
the basic tile $A$. The method is applicable in all dimensions $d$.

%

\begin{theorem}
\label{teo:mainteo} %
Let $V$ be a point lattice in $\R^d$ and let $\mu$ be a translation invariant measure on $\alf^V$. Let $(A,G)$ be a tiling
of $V$. Partition $A$ into $N$ nonempty sets $A_1,\dots, A_N$, with $1\leq N\leq |A|$. Define $A_{<k} = \bigcup_{i=1}^{k-1} A_i$.
Then
\begin{equation}\label{eq:latticeteo}
    s_\mu =  \frac{1}{|A|}\sum_{k=1}^N  S\left(A_k\Big|\left(A_k+G^-\right)
     \bigcup \left( A_{<k} + G \right)\right) \;. 
\end{equation}
\end{theorem}

\begin{proof}
Write $G = \langle \hat{e}_1,\dots,\hat{e}_d \rangle$. Define the subset $G_n$ of $G$ as
\begin{equation}\label{eq:GnDef}
G_n = \left\{ \sum_{i=1}^d g_i\hat{e}_i : |g_i| \leq n \;\forall\;
i \right\} \;.
\end{equation}
Let $\Lambda_n = A + G_n$. We claim that $\Lambda_n\uparrow \Z^d$ in the van Hove sense. The only non-trivial point to show is that
$\limn |\partial \Lambda_n|/|\Lambda_n|= 0 $. We start with a simple observation.
\begin{observation}\label{obs:G_nkonv}
For $G_n$ defined in \eqref{eq:GnDef}, and any $N\in\N$,
\[
\limn \frac{|G_{n+N}|}{|G_n|} = 1 \;.
\]
\end{observation}
\noindent
This follows since $|G_n| = (2n+1)^d$. Let $N$ be the smallest integer such that $U_d\subseteq \Lambda_N$. Such an integer must
exists since $(A,G)$ is a tiling of $V$. Note that by the definition of $G_n$ we have
\[
\Lambda_{Nm} = \Lambda_{N(m-1)}+\Lambda_{N} \;.
\]
In particular, this means that $\Lambda_{N(m-1)} + U_d \subseteq \Lambda_{Nm}$, so no element of $\Lambda_{N(m-1)}$ is in the
boundary of $\Lambda_{Nm}$. As a consequence, we have for any $n$ satisfying $N(m-1) < n \leq Nm$:
\begin{align*}
|\partial \Lambda_{n}| &\leq |\Lambda_{Nm}|-|\Lambda_{N(m-2)}| \,,
\end{align*}
while
\[
|\Lambda_{n}| > |\Lambda_{N(m-1)}| \;.
\]
Consequently, by Observation \ref{obs:G_nkonv}, $\limn |\partial \Lambda_n|/|\Lambda_n|  = 0$.

Now consider $S(\Lambda_n)$.
By using \eqref{eq:finitecondentropy} recursively we can write $S(\Lambda_n)$ in the form
\[
S(\Lambda_n) = \sum_{k=1}^N S\left(A_k + G_n \big| A_{<k} +
G_n\right).
\]
Consider an arbitrary term of the sum. Let $G_n$ inherit the lexicographic order from $G$. By using
\eqref{eq:finitecondentropy} again we obtain
\begin{equation}\label{eq:splitt2}
\begin{split}
S(A_k &  + G_n \big| A_{<k} + G_n) =  \\  &= \sum_{g\in G_n}
S\left(A_k+g
\big| \left(A_k + \{h\in G_n : h<g\} \right) \bigcup\left(A_{<k} + G_n\right) \right) = \\
&= \sum_{g\in G_n} S\left(A_k \big| \left(A_k + \{h\in (G_n-g) :
h<0\} \right) \bigcup\left(A_{<k} + G_n-g \right) \right) \,.
\end{split}
\end{equation}
Here we have used the translation invariance of the measure. By monotonicity of the entropy, Lemma \ref{lem:G_nkonv}, no term in
the sum can be smaller than $S(A_k|(A_k+G^-)\bigcup (A_{\leq k} + G))$. We can conclude that
\begin{equation}\label{S_larger_than}
S(A_k+G_n | A_{<k}+G_n ) \geq |G_n|\cdot
S\left(A_k|(A_k+G^-)\bigcup (A_{\leq k}+G)\right) \,.
\end{equation}
To show the opposite inequality, fix an $N\in\N$. For $n$ large, most of the terms in the sum on the right hand side of
\eqref{eq:splitt2} must be smaller than $S(A_k|(A_k+G_N^-)\cup (A_{<k}+G_N))$. Formally, if $g\in G_{n-N}$, then $G_N \subseteq
(G_n - g)$. Therefore
\begin{equation}
\begin{split}
\label{S_smaller_than} S(G_n+  A_k | G_n + A_{<k} ) \leq
&(|G_n|-|G_{n-N}|)|A_k|\log |\alf| \\ &+
|G_{n-N}|S(A_k|(A_k+G_N^-)\cup\left(A_{<k}+G_N\right) ).
\end{split}
\end{equation}
Since \eqref{S_smaller_than} is valid for all $N\geq 1$, Observation \ref{obs:G_nkonv} yields
\begin{align}
\limn \frac{1}{|\Lambda_n|} S(\Lambda_n) &=
\frac{1}{|A|}\sum_{k=1}^N \limn  \frac{1}{|G_n|} S\left(A_k+G_n
\big| A_{<k} + G_n \right) \nonumber \\ 
&\leq \frac{1}{|A|}\sum_{k=1}^N S(A_k|(A_k+G^-)\cap (A_{\leq k} +
G)) \,.
\end{align}
However, from \eqref{S_larger_than} we obtain the same inequality with $\geq$. The result follows.

\end{proof}

\noindent
Note that, the first term in \eqref{eq:latticeteo} is an entropy term that can be decomposed using the Theorem again, eventually resulting in a hierarchical decomposition of the lattice with corresponding entropy contributions. This is illustrated in the application to symbol sequences generated by cellular automata in Section 4.

\subsection{Examples}

\begin{example}
The well-known way of writing the entropy density of a spin system on $\Z^d$ as a conditional entropy relies on building up the
configuration layer by layer in succeeding dimensions. This corresponds to choosing $A = \{\underline{0}\}$, $A_1=A$ and
$G=V$, i.e., letting $G$ and $V$ have the same basis vectors.
\end{example}

\begin{example}
Take $V=\Z$. By putting $A=\{0,1\}$, $G=\langle 2\rangle$, $A_1=\{0\}$ and $A_2=\{1\}$, we first condition on every other spin to the left of position $0$,
and then on the entire sequence to left of position $1$ and every other spin to the right (starting with position $2$). Graphically, the expression is
\begin{equation}
s_\mu \hspace{2mm}= \hspace{2mm}
    \frac12\,\cdot\,
    \includegraphics[width=3.25cm]{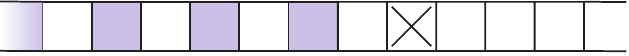} \;+\;
    \frac12\,\cdot\,
    \includegraphics[width=4cm]{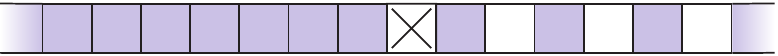}
\end{equation}
or, more formally, 
\begin{equation}
s_\mu=\frac12 \, \cdot \, S\big(\{0\} \big| \{0\}+G^-\big) + \frac12 \, \cdot \, S\big(\{1\} \big| (\{1\}+G^-)\cup(\{0\}+G)\big) \,.
\end{equation}
If $\mu$ is a Markov measure with memory one, this expression directly provides the relation $S(\{0\}|\{-1,1\}) =
2S(\{0\}|\{-1\}) - S(\{0\}|\{-2\})$.
\end{example}

\begin{example}
An analogue to the previous example for $V=\Z^2$ is to condition on a checker board pattern. This is obtained, e.g., by putting
$A=\{(0,0),(1,0)\}$, $A_1 = \{(0,0)\}$, $A_2 = \{(1,0)\}$ and $G = \langle(1,1),(1,-1)\rangle$. See Fig.~\ref{fig:2D_eks1} for an illustration.
\end{example}
\begin{figure}
\centering
\includegraphics[width=0.6\textwidth]{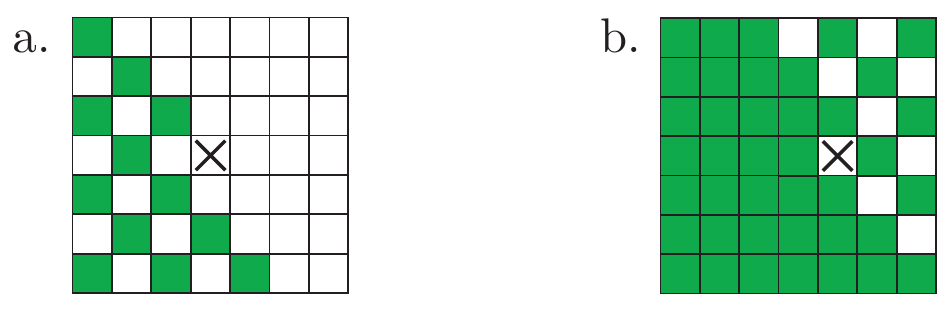}
\caption{The two components of the expression for $s_\mu$ in
$\Z^2$ found by using $A=\{(0,0),(1,0)\}$, $A_1 = \{(0,0)\}$, $A_2
= \{(1,0)\}$ and $G = \langle(1,1),(1,-1)\rangle$. The lattices stretch off to
infinity. \textbf{a.} $S(A_1|A_1+G^-)$. \textbf{b.}
$S(A_2|(A_2+G^-)\cup (A_1+G))$. }
\label{fig:2D_eks1}
\end{figure}

\begin{example}\label{ex:hex}
When $|A_k|>1$ the conditional entropy $S(A_k|B)$, for some $B\subseteq V$, can be written as a sum of $|A_k|$ conditional
entropies. We illustrate this using the hexagonal lattice $V = \langle(1,0),(\frac12,\frac{\sqrt {3}}{2})\rangle$. Put $A =
\{(-1,0),(0,0),(1,0)\}$ and $A_1 = A$. Let $G = \langle(3,0),(\frac{3}{2},\frac{\sqrt{3}}{2},)\rangle$. Then
\begin{align}
3s_\mu &= S(A|A+G^-) = \nonumber \\
&= S(\{(0,0)\}|A+G^-) +  S(\{(-1,0)\}|\{(0,0)\} \cup(A+G^-)) \nonumber  \\
& \;\; + S(\{(1,0)\}|\{(-1,0),(0,0)\} \cup(A+G^-)) \,.
\end{align}

\end{example}
\begin{figure}
\centering
\includegraphics[width=0.7\textwidth]{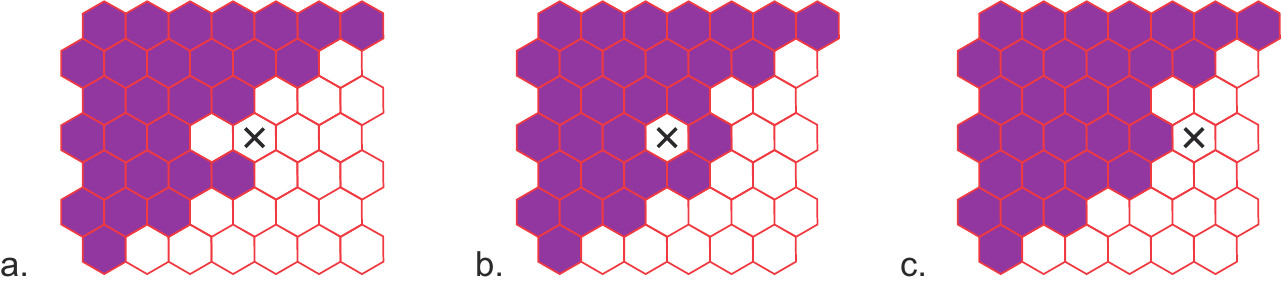}
\caption{The three components of the expression for $s_\mu$ of the hexagonal lattice constructed in example \ref{ex:hex}. Here, $s_\mu$ is
$\frac13$ times the sum of the displayed conditional entropies. }
\label{2D_hexagonal_eks}
\end{figure}

\subsection{Collections of translated point lattices}

Now we look at the case where we have a union of translations of a single point lattice, such as the triangular lattice in $\R^2$.

The lattice $V$ is then represented by $e_1,\dots,e_d$ and a finite set $T$ of translations $T=\{0,\tau_1,\dots, \tau_k\}$,
with $\tau_i\in\R^d-\langle e_1,\dots,e_d\rangle$. Note that $\underline{0}\in T$, and that no $\tau_i$ is in the point lattice $\sum v_i e_i$
defined by the $e_i$. Each $v\in V$ can be written as
\begin{equation}
v = \tau_v + \sum_{i=1}^d v_i e_i, \text{ where } \tau_v \in T,\,
v_i\in\Z \,.
\end{equation}
The representation in terms of $(v_1,\dots,v_d,\tau_v)$ is unique.

This union of point lattices can then be used in the same way as is expressed in Theorem \ref{teo:mainteo}, which implies that the tiling $(A,G)$ needs to be based on a finite subset $A$ that contains elements from all point lattices in the union forming $V$. Alternatively, one can derive the entropy separately for the different point lattices, each conditioned on the previously considered lattices. Both of these approaches are illustrated for a triangular lattice in the next section on Gibbs states. The triangular lattice is the union of two translations of the hexagonal lattice, each represented by $\langle(1,0),(\frac12,\frac{\sqrt {3}}{2})\rangle$, and with a translation set $T=\{(0,0),(\frac12,\frac{1}{2\sqrt{3}})\}$. 


\section{Applications to Gibbs states}
\label{Gibbs_states}

If the measure $\mu$ is a Gibbs measure corresponding to a potential $\Phi_X$ of finite interaction range $r$, then it has
the local Markov property. In particular, if only the nearest neighbours interact,  this is expressed as
\begin{equation}
\mu(x_{\Lambda}|x_{\Lambda^C}) =
\mu(x_{\Lambda}|x_{\partial(\Lambda^C)}) \,,
\end{equation}
for each finite $\Lambda$. I.e. when the spins at the outer boundary of $\Lambda$ is known, the spins located in the rest of
the lattice yield no further information about the spins in $\Lambda$. This property might not be valid for infinite
$\Lambda$. If it is, we say that the system has the global Markov property. This property can simplify the conditional entropy
expressions for the entropy density, since we can ignore all spins that are in the ''shadow'' of some specified spins, see, e.g., Goldstein et al \cite{goldstein90}.

For the standard tiling of the triangular lattice shown in Fig.~\ref{fig:2D_lattices}, the local Markov property would result in an entropy decomposition schematically shown in Fig.~\ref{fig:triang_lattice}. Here we choose the tiling $(A,G)$ with $A=\{(0,0),(\frac12,\frac{1}{2\sqrt{3}})\}$, picking one point from each of the two point lattices that form the triangular lattice. For $G$ we can then use the vectors of the hexagonal lattice, $\langle(1,0),(\frac12,\frac{\sqrt {3}}{2})\rangle$.

One way to derive the entropy would be to use the whole set of $A$ in the entropy density expression of Theorem \ref{teo:mainteo}, i.e.,
$s_\mu = \frac12 S(A | A + G^-)$. We can then split this conditional entropy into two terms, so that 
\begin{equation}
s_\mu = \frac12 S(A_1 | A + G^-) + \frac12 S(A_2 | (A + G^-) \cup A_1) \;, 
\end{equation}
where $A_1=\{(0,0)\}$ and $A_2=\{(\frac12,\frac{\sqrt {3}}{2})\}$.
This expression of the entropy density is shown in the first row of Fig.~\ref{fig:triang_lattice}, where we have also illustrated how the "shadowing" imposed by the global Markov property \cite{goldstein90} implies that there is no dependence on states in lower rows of the lattice.

Another way to derive the entropy using the same tiling is to partition $A$ into two sets $A_1$ and $A_2$, and then use the expression of Theorem \ref{teo:mainteo}, which yields
\begin{equation}
s_\mu = \frac12 S(A_1 | A_1 + G^-) + \frac12 S(A_2 | (A_2 + G^-) \cup (A_1 + G)) \;. 
\label{eq:triang_2}
\end{equation}
This is schematically illustrated in the second row of Fig.~\ref{fig:triang_lattice}, again making use of the "shadowing" of the lower rows. For the second conditional entropy term of Eq.~(\ref{eq:triang_2}), there is only dependence on the three nearest neighbours, resulting in $S(A_2 | \{ (0,0),e_1,e_2) \})$, with $e_1$ and $e_2$ being the vectors of the hexagonal lattice.


\begin{figure}
\centering
\includegraphics[width=0.7\textwidth]{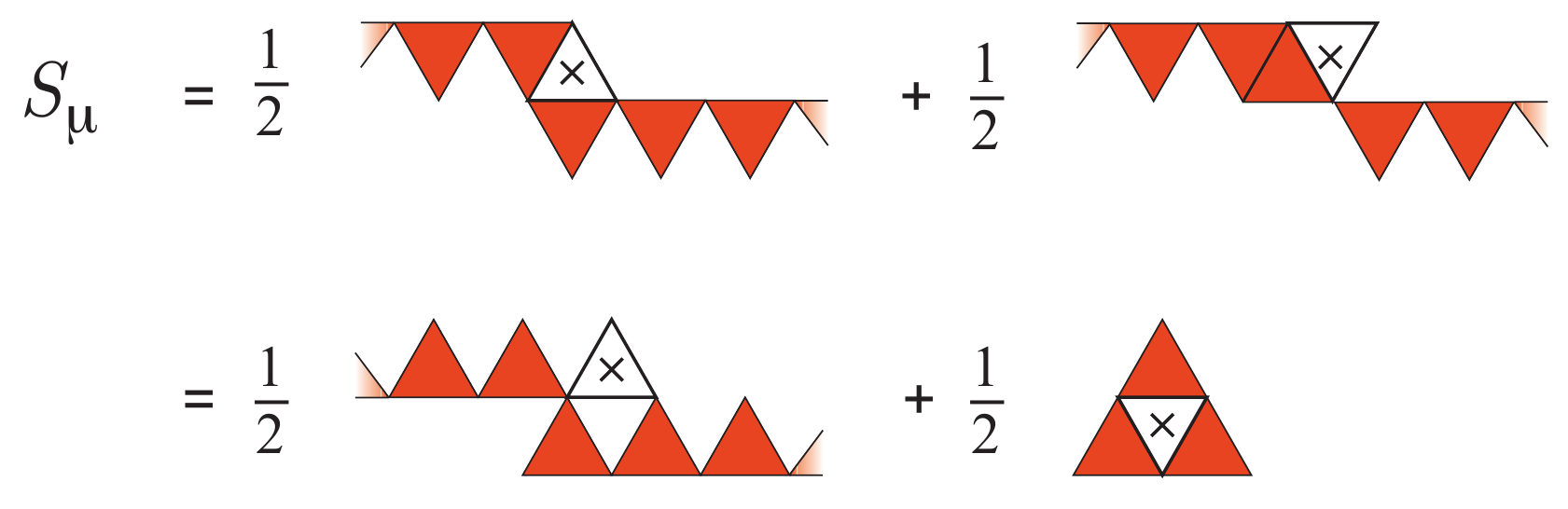}
\caption{An illustration of two ways to calculate the entropy of a nearest neighbor spin system on a triangular lattice. Each term
should be interpreted as the average local information of the cell containing the cross, conditioned on the states of the filled
cells.}
\label{fig:triang_lattice}
\end{figure}

\subsection{Application to Kikuchi's approximation of the 2D Ising model}
\noindent
In the early development of approximation schemes for the 2D Ising model, Kikuchi showed how a $2 \times 2$ block of states can be used in a good approximation of a measure on a 2D square lattice \cite{kikuchi51}. He showed that, for a rotation and translation invariant measure, the entropy density could well be approximated by 
\begin{equation} \label{eq:kikuchi-approx}
s_\mu \approx S(\{(0,0),(0,1),(1,0),(1,1)\})-2S(\{(0,0),(0,1)\})+S(\{(0,0)\}) \,.
\end{equation}
This expression was then used in minimisation of the free energy density, $f=u - k_\text{B} T s_\mu$, for the 2D Ising model, providing a good estimate of the free energy (with an error of less than $0.25\%$).

Using our scheme for calculating the entropy of a 2D square lattice, as was illustrated in Example 2.8, we can formulate a more efficient objective function for the free energy minimisation problem without increasing the number of variables. In a similar way as is illustrated in Fig.~\ref{fig:triang_lattice}, and following the notation of Example 2.8, we can write the entropy density in equilibrium as the sum of a checker board lattice entropy and a local conditional entropy, i.e.,
\begin{equation}
s_\mu= \frac12 \, S(\{0,0\} | \{0,0\}+G^- ) + \frac12 \, S( (1,0) | \{ (0,0),(1,1),(1,-1),(2,0)\} ) \,.
\end{equation}
Here we note that the first entropy can be approximated by the Kikuchi form of a $2 \times 2$ block of the checkerboard lattice, i.e., based on the conditional block 
in the second entropy term. The second entropy term is immediately determined by the Boltzmann probabilities, since the spin in question is shielded from the rest of the lattice and only depends on the four surrounding spin states. Therefore, we get the approximation
\begin{align} \label{eq:new-approx}
s_\mu \approx  &\frac12 \big( S(\{(0,0),(1,1),(1,-1),(2,0)\})-2S(\{(0,0),(1,1)\})+S(\{(0,0)\}) \big) + \nonumber \\
&+ \frac12 \, \sum_{x \in C} \mu(x) S_\text{B}(P_Z(x))  \,,
\end{align}
where summation is over all configurations $C$ on the diagonal $2 \times 2$ block, which leads to an average entropy of the conditional distribution of a spin $Z$ surrounded by such a block $x$. Here, $S_\text{B}(\{p,1-p\})=-p \log p-(1-p) \log(1-p)$. Since the second term is determined by the Boltzmann distribution, $P_Z$, it will only depend on the probability distribution over the diagonal $2 \times 2$ block, and thus the approximation (\ref{eq:new-approx}) will depend on the same number of variables as the original Kikuchi approximation (\ref{eq:kikuchi-approx}). Numerical calculations of the two approximations as functions of temperature is shown in Figure~\ref{fig:kikuchi}, and this illustrates that our scheme captures a little more of the correlations affecting the entropy and thus improves on the free energy estimate. When we are close to the critical temperature, the importance of longer correlations increases and the approximations are less good.

\begin{figure}
\centering
\includegraphics[width=0.85\textwidth]{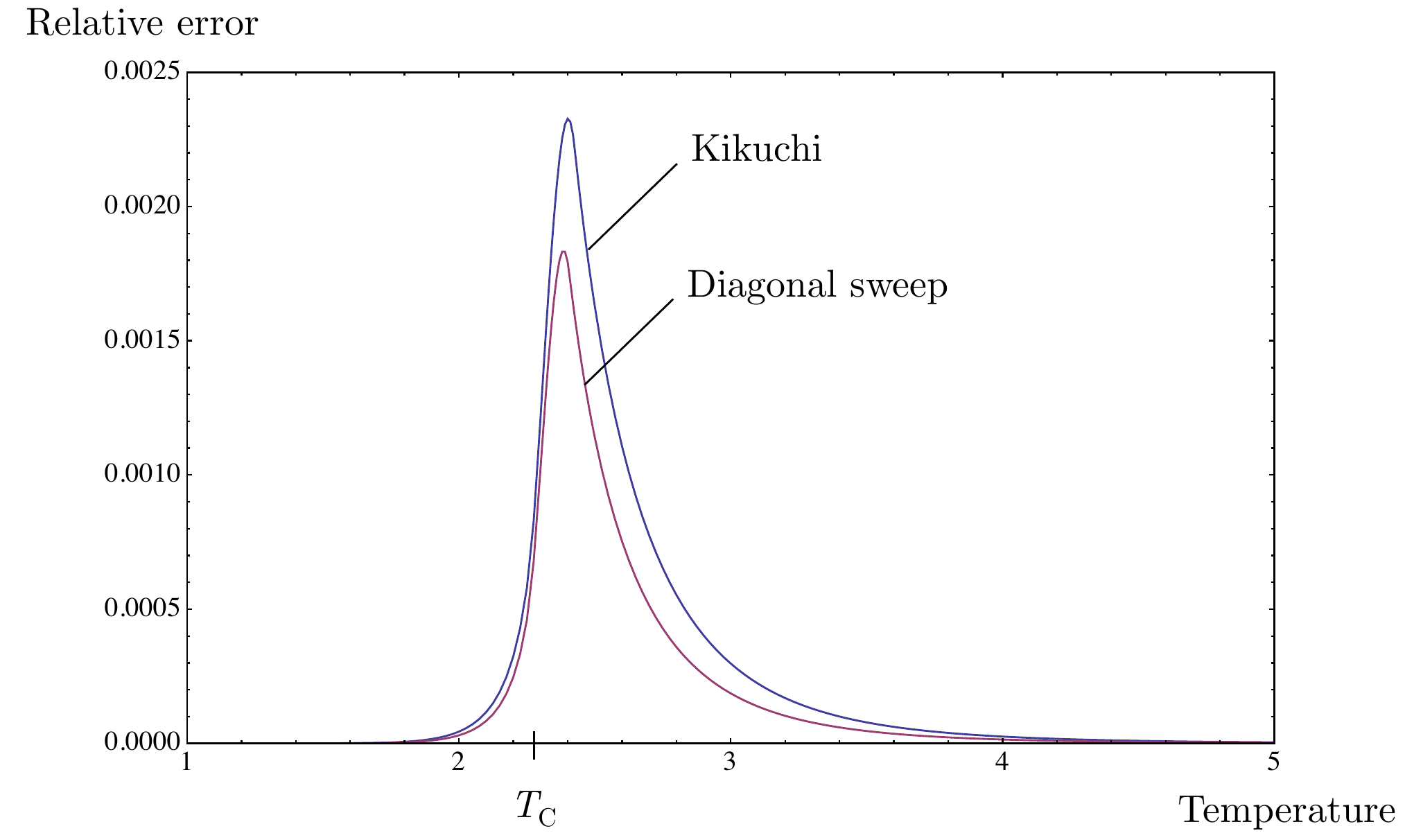}
\caption{The error of the free energy estimate for the Kikuchi approximation and for the modification based on our approach. }
\label{fig:kikuchi}
\end{figure}

\section{Applications to states generated in the time evolution of cellular automata}
\label{sec:CA}
The time evolution of cellular automaton states, each being a configuration of local states in an infinite lattice, provides good examples of systems where correlations are built up in a complex way, even if we start from an initial state generated by a simple stochastic process (usually without any correlations). Here we illustrate how the formalism can be applied to give better estimates on the entropy density of the states in cellular automata time evolution, at least for some points in time.

\subsection{A hierarchical scheme for decomposing the entropy}
By repeated use of Theorem~\ref{teo:mainteo}, we construct a hierarchical decomposition of the entropy density, which turns out to be highly efficient for states generated at certain time steps of certain cellular automaton rules.

The scheme we use is one where the first sweep involves states separated by a distance $2^m$, for some $m>1$. In the next step, the states in the middle between two consecutive states from the previous sweep are considered, and this is then repeated until all states of the lattice have been covered. For $m=3$, this can be graphically represented as
\begin{align} \label{eq:ca-example}
s_\mu \hspace{2mm}= \hspace{2mm}
    \frac18\,\cdot\,
   &\includegraphics[width=6cm]{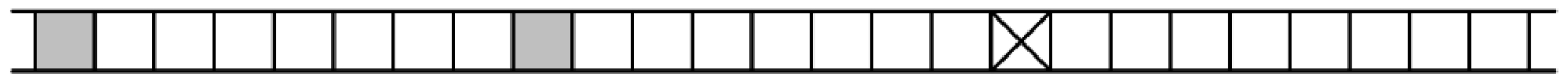} \;+ \nonumber \\
    + \; \frac18\,\cdot\,
   &\includegraphics[width=6cm]{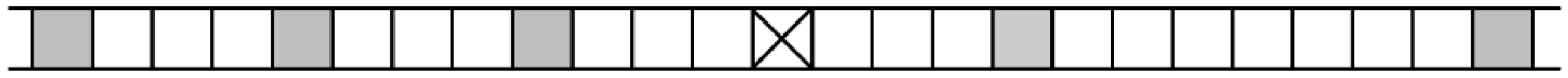} \;+ \nonumber \\
    + \; \frac14\,\cdot\,
   &\includegraphics[width=6cm]{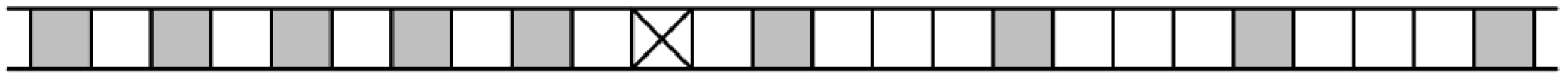} \;+ \nonumber \\
    + \; \frac12\,\cdot\,
   &\includegraphics[width=6cm]{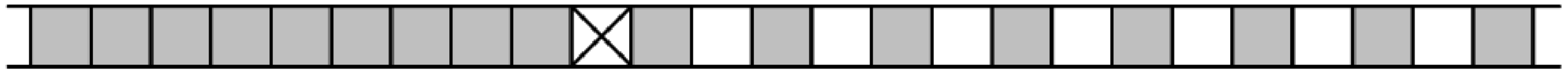}
\end{align}
This decomposition of the entropy is obtained by repeatedly using Theorem~\ref{teo:mainteo} on the first term in \eqref{eq:latticeteo}. With $m$ hierarchical levels, each one determined by a pair of points, $A_k=\{0,2^{k-1}\}$, with $1\le k \le m$, and $G_k=\langle2^k\rangle$, we get
\begin{align} \label{eq:hierarchical-ent}
s_\mu = \, &\frac{1}{2^m} S\big(\{0\} \big| \{0\}+G_m^- \big) +  \nonumber \\
&+ \sum_{k=1}^m \frac{1}{2^k} S\big( \{2^{k-1}\} \big| (\{2^{k-1}\}+G_k^-)\cup (\{0\}+G_k) \big) \,.
\end{align}
For $m=3$ this results in the entropy decomposition graphically illustrated in \eqref{eq:ca-example}. Here we note that 
\begin{align} \label{eq:cond-ent-ineq}
S\big(\{0\} \big| \{0\}+G_m^- \big) \ge S\big( \{2^{k-1}\} \big| (\{2^{k-1}\}+G_k^-)\cup (\{0\}+G_k) \big) \, ,
\end{align}
for any $k \le m$, since, as expressed by Lemma \ref{lem:G_nkonv}, the conditional entropy on the lefthand side depends on states that are a subset of the states that any of the entropies on the righthand side depends on. This may be a useful relation in numerical estimates of the entropy when long-range correlations dominate. 

For numerical estimates of the entropy, one typically calculates an estimate of the conditional entropy terms as a converging series with increasing size of the conditional configuration. In the examples below, we let $K$ denote the number of states in the conditional configuration. The $K$'th estimate to the entropy $s_\mu$ is therefore written
\begin{align}
s_\mu \le \frac{1}{2^m} S_{0,m}^{(K)} + \sum_{k=1}^m \frac{1}{2^k} S_{1,k}^{(K)} \,,
\end{align}
where we have introduced the notation $S_{0,m}^{(K)}$ for the $K$'th estimate of the conditional entropy in the first term in \eqref{eq:hierarchical-ent}, and $S_{1,m}^{(K)}$ for the corresponding $K$'th estimates of the conditional entropies in the sum. For the $K$ states in the conditional configuration we have chosen those that are closest to the position of the actual state (selecting a position to the left before one to the right). Since $S_{0,m}^{(K)} \ge S_{0,m}^{(\infty)} \ge S_{1,m}^{(\infty)} $, a stronger inequality, and in some situations a better estimate of $s_\mu$ is given by
\begin{align} \label{eq:entropy-estimate}
s_\mu \le \frac{1}{2^m} S_{0,m}^{(K)} + \sum_{k=1}^m \frac{1}{2^k} \min\big( S_{0,m}^{(K)}, S_{1,k}^{(K)} \big) \,.
\end{align}
This serves as a more efficient estimate whenever correlations that are multiples of $m$ dominate (i.e., if $S_{0,m}^{(K)} < S_{1,k}^{(K)}$), since all states in the conditional configuration of the $K$ states in $S_{0,m}$ are contributing to the correlations. The estimate \eqref{eq:entropy-estimate} is used in following the numerical example.

\subsection{Application to elementary cellular automaton rule 60}
The elementary cellular automaton rule 60 is an example of a dynamical system in which correlations are being spread out on increasing distances, provided that the starting state is not fully random. At each time step, we can characterise the state -- an infinite sequence of lattice sites, each with a state $0$ or $1$ -- by the entropy of the lattice. The rule replaces the local lattice state ($0$ or $1$) with its sum modulo $2$ with its left neighbour. This implies that the rule is surjective, i.e., it possesses a certain degree of reversibility, so that the entropy of the state is conserved in the time evolution. But, as correlation information (redundancy) is spread out over ever increasing block lengths, it becomes increasingly more difficult to fully detect the entropy in the state \cite{lindgren87, lindgrennordahl}.

If the initial configuration is given by a Bernoulli distribution, say with $P(1)=p$, then, even if the entropy stays constant, there is a highly complex dynamics in terms of how correlations are spread out in the system. It is known that the average correlation length increases linearly in time \cite{lindgrennordahl}, when measured by how the block entropies, in the traditional approach, converge to the entropy density of the configuration. This is a strong indication on that it becomes increasingly more difficult to get a numerical estimate of the entropy of the cellular automaton state configuration as time increases. Therefore, we explore how the hierarchical entropy estimate \eqref{eq:entropy-estimate} performs in comparison with the traditional approach.

\begin{figure}
\centering
\includegraphics[width=0.85\textwidth]{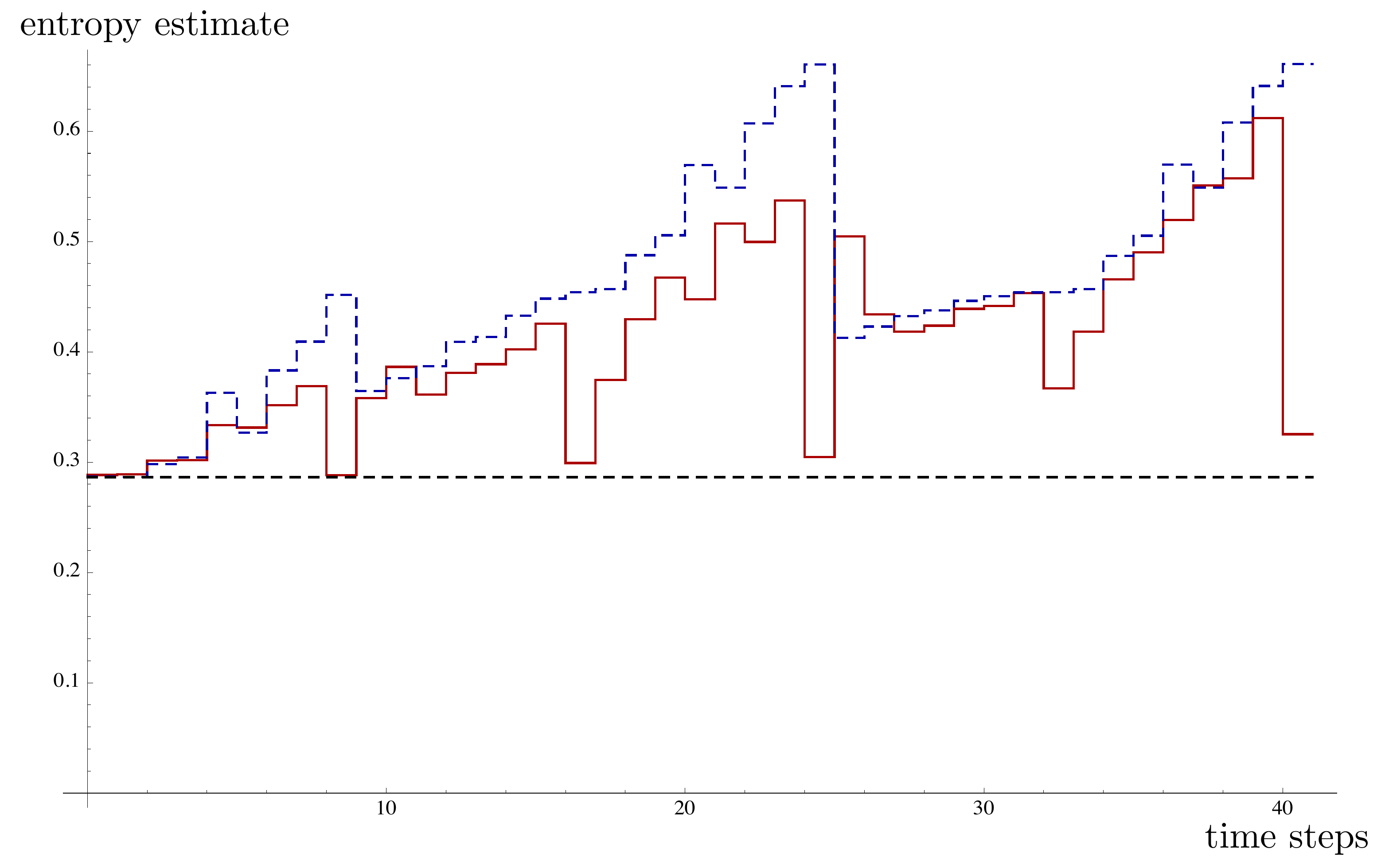}
\caption{The entropy estimates are shown for the 40 first time steps $(t=0,1,2,...,40)$ of cellular automaton rule 60, using both the traditional sequential estimate (blue dashed line) and the hierarchical estimate \eqref{eq:entropy-estimate} (as the solid red line), with the exact entropy value as the horizontal dashed line. The initial state is characterised by a Bernoulli distribution with $p=0.05$ implying an entropy density of $s \approx 0.286 \log(2)$ for all time steps. The numerical estimates are based on a cellular automaton with configuration size of $2^{16}$ cells.}
\label{fig:entropy-comp-fig}
\end{figure}

Fig.~\ref{fig:entropy-comp-fig} shows, as a function of time, the entropy estimates based on $K=7$ states in the conditional configurations, for the traditional approach and for the hierarchical scheme \eqref{eq:entropy-estimate}, respectively. The graph illustrates that for most time steps the estimates are not very close to the exact value, but also that for most time steps the hierarchical estimate works better. It is clear that, at certain time steps, the lattice configurations do have structure that is well captured by the hierarchical approach. This happens, for example, at $t=8, 16, 24, 32, \text{and } 40$, which are time steps where correlation of length $8$ is dominating, a consequence of the cellular automaton rule 60 characteristics. 

Time steps that are powers of $2$, i.e., $t = 2^k$, with $k=1,2,3,...$, have a very characteristic correlation structure. At any such time step, all correlations are multiples of $2^k$, i.e., they stretch over blocks of length $m \, 2^k + 1$, with $m \ge 1$. This means that the state is decomposed into $2^k$ different independent subsystems, with no correlations in between them. The traditional approach in estimating the entropy, by extending the conditional block one state by one, will quickly run into numerical difficulties when $k$ increases. This is clearly seen in Fig.~(\ref{fig:entropy-comp-fig}), even though we note that at $t=32$, also the hierarchical approach has an error in its estimate. Here we have used a first distance of $8$ in top level of the hierarchical scheme, see Eq.~(\ref{eq:ca-example}), but if we instead would have chosen $32$ as the starting distance, we would capture the underlying entropy to much a higher degree of accuracy.

\section{Conclusions}
The main contribution with this paper is Theorem \ref{teo:mainteo} showing that there is high flexibility in how one can search for and calculate correlations that determines the entropy density of a lattice system. The traditional approach, exemplified in one dimension, based on sequentially increasing blocks of local states by one state at the time, runs into an exponential explosion of terms to calculate if correlations are long. The flexibility that is established by Theorem \ref{teo:mainteo}, indicates that there may exist more efficient ways in which the lattice may be decomposed, which means that one does not necessarily have to sweep the lattice moving from one site to the neighbouring one. Instead, on may take larger steps, and then in following sweeps one can take the remaining sites into account, so that, at the end, the full lattice has been covered.

By repeated use of Theorem \ref{teo:mainteo}, we note that a hierarchical scheme for the decomposition of the lattice can be obtained. This was illustrated with the application to a surjective cellular automaton rule, in which the entropy is conserved but where information in correlations may still be spread out on ever increasing distances. The numerical estimates of the entropy density based on a hierarchical scheme, compared with the traditional approach, illustrates that for most time steps the hierarchical scheme yields a better estimate, and that, for some time steps, it is capable of capturing the entropy with high accuracy.

The two-dimensional illustration, using the Ising model and the Kikuchi approximation of the entropy estimate, shows that there is a natural way to decompose the two-dimensional lattice of a spin system when we have nearest neighbour interactions only. In general, we are interested in the equilibrium description of the system, and we thus are interested in minimising the free energy. In a first sweep, one only considers the states in a checkerboard pattern of the lattice, thus covering half of the system. In this sub-lattice there is no energy interaction, but only contributions to the entropy. In the second sweep, the remaining states are filled in, but here we note that in equilibrium their statistics is fully determined by the local Boltzmann distribution, as each of these states will only depend on the four neighbours from the checkerboard configuration. This second term thus determines the energy, and half of the entropy contribution. As was shown in Figure~\ref{fig:kikuchi}, this gives a better estimate of the free energy when using the same number of free variables in the minimisation procedure as is done in the Kikuchi approximation.

Whether a system is to be considered complex or not, may certainly depend on how one looks for correlations. The most commonly used complexity quantity for one-dimensional systems is the effective measure complexity \cite{grassberger86}, or the excess entropy (see \cite{crutchfield03} and references therein), which implicitly is based on a sequential, site-by-site, extension of blocks. This is also the reason why this complexity quantity increases linearly in the time evolution of cellular automaton rule 60, despite the fact that at certain points in time, the system is not that complex if one would look for correlations in a different way. Whether a more general complexity measure could be constructed based on the flexibility on how lattices can be decomposed, as stated by Theorem \ref{teo:mainteo}, is a question for future research.

\section{Acknowledgements}
Financial support from EU-FP7 project MatheMACS is gratefully acknowledged.


\bibliographystyle{plain}
\bibliography{ref-entropy}

\begin{thebibliography}{10}

\bibitem{alexandrowicz71}
Z.~Alexandrowicz.
\newblock Stochastic models for the statistical description of lattice systems.
\newblock {\em J. Chem. Phys.}, 55(6):2765--2779, 1971.

\bibitem{binder85}
K.~Binder.
\newblock The {M}onte {C}arlo method for the study of phase-transitions - a
  review of some recent progress.
\newblock {\em J. Comp. Phys.}, 59(1):1--55, 1985.

\bibitem{coverthomas}
T.~M. Cover and J.~A. Thomas.
\newblock {\em Elements of information theory}.
\newblock Wiley Series in Telecommunications. John Wiley \& Sons Inc., New
  York, 1991.

\bibitem{crutchfield03}
D.~P. Feldman and J.~P. Crutchfield.
\newblock Structural information in two-dimensional patterns: Entropy
  convergence and excess entropy.
\newblock {\em Phys. Rev. E}, 67, 2003.
\newblock 051104.

\bibitem{goldstein90}
S.~Goldstein, R.~Kuik, and A.~G. Schlijper.
\newblock Entropy and global {M}arkov properties.
\newblock {\em Comm. Math. Phys.}, 126(3):469--482, 1990.

\bibitem{grassberger86}
P.~Grassberger.
\newblock Toward a quantitative theory of self-generated complexity.
\newblock {\em Int. J. Theor. Phys.}, 25:907--938, 1986.

\bibitem{helvik-et-al}
T.~Helvik, K.~Lindgren, and M.G. Nordahl.
\newblock Continuity of information transport in surjective cellular automata.
\newblock {\em Communications in Mathematical Physics}, 272:53--74, 2007.

\bibitem{kikuchi51}
R.~Kikuchi.
\newblock A theory of cooperative phenomena.
\newblock {\em Phys. Rev.}, 81:988--1003, 1951.

\bibitem{kramers41}
H.~A. Kramers and G.~H. Wannier.
\newblock Statistics of the two-dimensional ferromagnet.
\newblock {\em Phys. Rev.}, 60:252--262, 1941.

\bibitem{kullback51}
S.~Kullback and R.~A. Leibler.
\newblock On information and sufficiency.
\newblock {\em Ann. Math. Stat.}, 22(1):79--86, 1951.

\bibitem{lindgren87}
K.~Lindgren.
\newblock Correlations and random information in cellular automata.
\newblock {\em Complex Systems}, 1:529--543, 1987.

\bibitem{lindgrennordahl}
K.~Lindgren and M.~G. Nordahl.
\newblock Complexity measures and cellular automata.
\newblock {\em Complex Systems}, 2:409--440, 1988.

\bibitem{Marcelja}
S.~Marcelja.
\newblock Entropy of phase-separated structures.
\newblock {\em Physica A}, 231:168--177, 1996.

\bibitem{meirovitch77}
H.~Meirovitch.
\newblock Calculation of entropy with computer simulation methods.
\newblock {\em Chem. Phys. Let.}, 45(2):389--392, 1977.

\bibitem{Meirovitch83}
H.~Meirovitch.
\newblock Methods for estimating entropy with computer-simulation - the simple
  cubic {I}sing lattice.
\newblock {\em J. Phys. A}, 16:839--848, 1983.

\bibitem{Meirovitch83b}
H.~Meirovitch.
\newblock A monte carlo study of the entropy, the pressure, and the critical
  behavior of the hard-square lattice gas.
\newblock {\em J. Stat. Phys.}, 30:681--698, 1983.

\bibitem{Meirovitch99}
H.~Meirovitch.
\newblock Simulation of a free energy upper bound, based on the anticorrelation
  between an approximate free energy functional and its fluctuation.
\newblock {\em J. Chem. Phys}, 111(16):7215--7224, 1999.

\bibitem{Olbrich00}
E.~Olbrich, R.~Hegger, and H.~Kantz.
\newblock Local estimates for entropy densities in coupled map lattices.
\newblock {\em Phys. Rev. Lett.}, 84:2132--5, 2000.

\bibitem{onsager}
L~Onsager.
\newblock Crystal statistics. {I}. a two-dimensional model with a
  order-disorder transition.
\newblock {\em Phys. Rev}, 65:117--149, 1944.

\bibitem{schlijper83}
A.~G. Schlijper.
\newblock Convergence of the cluster-variation method in the thermodynamic
  limit.
\newblock {\em Phys. Rev. B}, 27:6841, 1983.

\bibitem{schlijper89}
A.~G. Schlijper and B.~Smit.
\newblock Two-sided bounds on the free energy from local states in {M}onte
  {C}arlo simulations.
\newblock {\em J. Stat. Phys.}, 56(3/4):247, 1989.

\bibitem{schlijper90}
A.~G. Schlijper, A.~R.~D. {van Bergen}, and B.~Smit.
\newblock Local-states method for the calculation of free energies in {M}onte
  {C}arlo simulations of lattice models.
\newblock {\em Phys Rev. A}, 41(2):1175, 1990.

\bibitem{Shannon48}
C.~E. Shannon.
\newblock A mathematical theory of communication.
\newblock {\em Bell System Tech. J.}, 27:379--423, 623--656, 1948.

\end{thebibliography}

\end{document}